\documentclass[a4paper,UKenglish]{lipics-v2016}

\usepackage{microtype,xspace,stmaryrd}

\newcommand{\NP}{{\sf NP}\xspace}

\newcommand{\FPT}{{\sf FPT}\xspace}
\newcommand{\ETH}{{\sf ETH}\xspace}
\newcommand{\SETH}{{\sf SETH}\xspace}

\newcommand{\Acal}{\mathcal{A}}

\newcommand{\Rbb}{\mathbb{R}}

\newcommand{\W}{{\sf W}\xspace}

\newcommand{\poly}{{\sf P}\xspace}

\theoremstyle{plain}

\newtheorem{proposition}[theorem]{Proposition}

\newenvironment{proof-claim}[1][]{\par \noindent {\textcolor{darkgray}{\sffamily\bfseries Proof of the claim.#1}}\ }{\hfill$\Box$\bigskip}

\newenvironment{proof-sketch}[1][]{\par \noindent {\textcolor{darkgray}{\sffamily\bfseries Sketch of proof.#1}}\ }{\hfill$\Box$\bigskip}

\newenvironment{proof-DAG}[1][]{\par \noindent {\textcolor{darkgray}{\sffamily\bfseries Proof of Theorem~\ref{thm:algoDAG}.#1}}\ }{\hfill$\Box$\bigskip}

\usepackage{cite}
\usepackage{hyperref}
\usepackage{boxedminipage}

\newcommand{\Ocal}{\mathcal{O}}

\newcommand{\np}{{\sf NP}\xspace}

\newcommand{\fpt}{{\sf FPT}\xspace}

\newcommand{\igjournal}[1]{}

\newcommand{\paraprobl}[4]
{
  \begin{flushleft}
    \fbox{
      \begin{minipage}{13.5cm}
        \noindent {\textsc {#1}}\\
        {\bf Input:} #2\\
        {\bf Parameter:} #4\\
        {\bf Question:} #3
      \end{minipage}
    }
  \end{flushleft}
}

\newtheorem{claim}{Claim}

\title{Dual parameterization of Weighted Coloring\footnote{Emails of authors: \texttt{julio@mat.ufc.br}, \texttt{campos@lia.ufc.br}, \texttt{cvinicius05@gmail.com}, \texttt{viniciussantos@dcc.ufmg.br}, \texttt{ignasi.sau@lirmm.fr}, \texttt{anasilva@mat.ufc.br}. \newline Work supported by French ANR projects DEMOGRAPH (ANR-16-CE40-0028) and ESIGMA (ANR-17-CE40-0028), and by Brazilian projects CNPq 306262/2014-2, CNPq 421660/2016-3,  CNPq 401519/2016-3, FAPEMIG,  and Funcap PNE-0112-00061.01.00/16.}
}

\titlerunning{Dual parameterization of Weighted Coloring}

\author[1]{J\'ulio Ara\'ujo}
\author[2]{Victor A. Campos}
\author[3]{Carlos Vin\'icius G. C. Lima}
\author[3]{Vin\'icius Fernandes dos Santos}
	\author[4]{Ignasi Sau}
\author[1]{Ana Silva}

\affil[1]{Departamento de Matem\'atica, ParGO research group, Universidade Federal do Ceará, Fortaleza, Brazil}

\affil[2]{Departamento de Computa\c{c}\~{a}o, ParGO research group, Universidade Federal do Ceará, Fortaleza, Brazil}

\affil[3]{Departamento de Ci\^encia da Computa\c{c}\~{a}o,
Universidade Federal de Minas Gerais, Belo Horizonte, Brazil}

\affil[4]{CNRS, AlGCo project team, LIRMM, Universit\'e de Montpellier, France}

\authorrunning{J. Ara\'ujo, V. A. Campos, C. V. G. C. Lima, V. F. dos Santos, I. Sau, and A. Silva} 

\Copyright{J\'ulio Ara\'ujo, Victor A. Campos, Carlos V. Lima, Vin\'icius F. dos Santos, Ignasi Sau, Ana Silva}


\EventEditors{John Q. Open and Joan R. Access}
\EventNoEds{2}
\EventLongTitle{}
\EventShortTitle{}
\EventAcronym{}
\EventShortTitle{}
\EventAcronym{}
\EventDate{}
\EventLocation{}
\EventLogo{}

\begin{document}

\maketitle


\begin{abstract}
Given a graph $G$, a \emph{proper $k$-coloring} of $G$ is a partition $c = (S_i)_{i\in [1,k]}$ of $V(G)$ into $k$ stable sets $S_1,\ldots, S_{k}$. Given a weight function $w: V(G) \to \Rbb^+$, the \emph{weight of a color} $S_i$ is defined as $w(i) = \max_{v \in S_i} w(v)$ and the \emph{weight of a coloring} $c$  as $w(c) = \sum_{i=1}^{k}w(i)$.  Guan and Zhu [Inf. Process. Lett., 1997] defined the \emph{weighted chromatic number} of a pair $(G,w)$, denoted by $\sigma(G,w)$, as the minimum weight of a proper coloring of $G$. The problem of determining $\sigma(G,w)$ has received considerable attention during the last years, and has been proved to be notoriously hard: for instance, it is \np-hard on split graphs, unsolvable on $n$-vertex trees in time $n^{o(\log n)}$ unless the \ETH fails, and {\sf W}[1]-hard on forests parameterized by the size of a largest tree. In this article we provide some positive results for the problem, by considering its so-called \emph{dual parameterization}: given a vertex-weighted graph $(G,w)$ and an integer $k$, the question is whether $\sigma(G,w) \leq \sum_{v \in V(G)} w(v) - k$. We prove that this problem is \FPT by providing an algorithm running in time $9^k \cdot n^{\Ocal(1)}$, and it is easy to see that no algorithm in time $2^{o(k)} \cdot n^{\Ocal(1)}$ exists under the \ETH. On the other hand, we present a kernel with at most $(2^{k-1}+1) (k-1)$ vertices, and we rule out the existence of polynomial kernels unless ${\sf NP} \subseteq {\sf coNP} / {\sf poly}$, even on split graphs with only two different weights. Finally, we identify some classes of graphs on which the problem admits a polynomial kernel, in particular interval graphs and subclasses of split graphs, and in the latter case we present lower bounds on the degrees of the polynomials.
\end{abstract}

\section{Introduction}
\label{sec:intro}

A \emph{(vertex) $k$-coloring} of a graph $G = (V,E)$ is a function $c:V(G)\to \{1,\ldots, k\}$. Such coloring $c$ is \emph{proper} if $c(u)\neq c(v)$ for every edge $\{u,v\} \in E(G)$. All the colorings we consider in this paper are proper, hence we may omit the word ``proper''. The \emph{chromatic number} $\chi(G)$ of $G$ is the minimum integer $k$ such that $G$ admits a $k$-coloring. Given a graph $G$, determining $\chi(G)$ is the goal of the classical \textsc{Vertex Coloring} problem.
If $c$ is a $k$-coloring of $G$, then $S_i = \{u\in V(G)\mid c(u) = i\}$ is a stable (or independent) set. 
With slight abuse of notation, we shall also call such a set $S_i$ a \emph{color}. 

In this paper we study a generalization of \textsc{Vertex Coloring} for vertex-weighted graphs that has been defined by Guan and Zhu~\cite{GZ97}. Given a graph $G$ and a weight function $w: V(G) \to \Rbb^+$, 
the \emph{weight of a color} $S_i$ is defined as $w(i) = \max_{v \in S_i} w(v)$. Then, the \emph{weight of a coloring} $c$ is $w(c) = \sum_{i=1}^{k}w(i)$. In the \textsc{Weighted Coloring} problem, the goal is to determine the \emph{weighted chromatic number} of a pair $(G,w)$, denoted by $\sigma(G,w)$, which is the minimum weight of a coloring of $(G,w)$. A coloring $c$ of $G$ such that $w(c) = \sigma(G,w)$ is an \emph{optimal weighted coloring}. Guan and Zhu~\cite{GZ97} also defined, for a positive integer $r$, $\sigma(G,w;r)$ as the minimum of $w(c)$ among all $r$-colorings $c$ of $G$, or as $+\infty$ if no $r$-coloring exists.
Note that $\sigma(G,w) = \min_{r \geq 1}\sigma(G,w;r)$.  It is worth mentioning that the \textsc{Weighted Coloring} problem is also sometimes called \textsc{Max-Coloring} in the literature; see for instance~\cite{KaMe09,PemmarajuPR05}.
Guan and Zhu defined this problem in order to study practical applications related to resource allocation, which they describe in detail in~\cite{GZ97}.
One should observe that if all the vertex weights are equal to one, then $\sigma(G,w) = \chi(G)$, for every graph $G$. Consequently, determining $\sigma(G,w)$
is an {\sf NP}-hard problem on general
graphs~\cite{Karp72}. In fact, this problem has been shown to be {\sf NP}-hard even on very restricted graph classes, such as split graphs with only two different weights, interval graphs, triangle-free planar graphs with bounded degree, and bipartite graphs~\cite{DWMP02, WDEMP09, EMP06}. On the other hand, the weighted chromatic
number of cographs and of some subclasses of bipartite graphs can be found in
polynomial time~\cite{DWMP02, WDEMP09}.

The complexity of  \textsc{Weighted Coloring} on trees (and forests) has
attracted considerable attention in the literature. Guan and Zhu~\cite{GZ97} left as an open problem whether \textsc{Weighted Coloring} is polynomial on trees and, more generally, on graphs of bounded treewidth. Escoffier et al.~\cite{EMP06} found a polynomial-time approximation scheme to solve \textsc{Weighted Coloring} on bounded treewidth graphs, and
Kavitha and Mestre~\cite{KaMe09} showed that the problem is in \poly on the class of trees where vertices with degree at least three induce a stable set. But the question of Guan and Zhu has been answered only recently, when  Ara\'ujo et al.~\cite{Julio} showed that, unless the Exponential Time Hypothesis (\ETH)\footnote{The \ETH states that 3-\textsc{SAT} cannot be solved in subexponential
time; see~\cite{ImpagliazzoPZ01,ImpagliazzoP01} for more details.} fails, there is no algorithm computing the weighted chromatic number of $n$-vertex trees in time $n^{o(\log n)}$. Moreover, as discussed in~\cite{Julio}, this lower bound is tight. Very recently,
Ara\'ujo et al.~\cite{AraujoBS18} focused on the parameterized complexity of computing $\sigma(G,w)$ and $\sigma(G,w;r)$ when $G$ is a forest,
and they proved that computing $\sigma(G,w)$ is {\sf W}[1]-hard parameterized by the size of a largest  tree of $G$, and that computing $\sigma(G,w;r)$  is {\sf W}[2]-hard parameterized by $r$.

\smallskip

In view of the above discussion, we can conclude that \textsc{Weighted Coloring} is a particularly hard problem from a computational point of view, and that the positive results in the literature are quite scarce. To shed some light on this direction, in this paper we adopt the perspective of parameterized complexity and consider the \emph{dual parameterization} of the problem, which we call \textsc{Dual Weighted Coloring} and is formally defined as follows:


\paraprobl{Dual Weighted Coloring}
{A vertex-weighted graph $(G,w)$ and a positive integer $k$.}
{Is  $\sigma(G,w) \leq \sum_{v \in V(G)}w(v) - k$?}
{$k$.}


Since by definition of parameterized problem (cf.~\cite{DF13,CyganFKLMPPS15}) the parameter needs to be a  non-negative integer, for the above parameterization to make sense we will assume henceforth that all vertex-weights are positive integers. We will denote by $n$ the number of vertices of the input graph of \textsc{Dual Weighted Coloring}.

The motivation for considering such parameterization is to take as the parameter the ``savings'' with respect to the trivial upper bound of $\sum_{v \in V(G)}w(v)$ on $\sigma(G,w)$. This approach has proved to be very useful for the classical \textsc{Vertex Coloring} problem, especially from the approximation point of view~\cite{DemangeGP94,HassinL94,Halldorsson95,Halldorsson96,DuhF97}. From a parameterized perspective, Chor et al.~\cite{ChorFJ04} presented an \fpt algorithm for \textsc{Dual Vertex Coloring}. Concerning kernelization, it is not difficult to see~\cite[Exercise 2.22]{CyganFKLMPPS15} that the problem admits a kernel with at most $3k$ vertices, by applying the so-called \emph{crown reduction rule} to the complement of the input graph $G$. Other results concerning dual parameterization are, for instance, \fpt algorithms for the Grundy and $b$-chromatic numbers of a graph~\cite{HavetS13}, the parameterized approximability of subset graph problems~\cite{BonnetP17}, or the existence of polynomial kernels for the \textsc{Set Cover} and \textsc{Hitting Set} problems~\cite{GutinJY11,BasavarajuFRS16}.



\medskip

\noindent\textbf{Our results.} Our first result is an \FPT algorithm for \textsc{Dual Weighted Coloring} running in time $9^k \cdot n^{\Ocal(1)}$,  based on some standard dynamic programming ideas used in coloring and partition problems (see, for example,~\cite{Law76} and~\cite[Section 3.1.2]{FK10}).
It is easy to see that a subexponential algorithm is unlikely to exist. Indeed, consider the \textsc{$3$-Coloring} problem, which corresponds to the unweighted version of \textsc{Dual Weighted Coloring} with parameter $k = n-3$. Since \textsc{$3$-Coloring} cannot be solved in time $2^{o(n)}$ under the \ETH~\cite{LokshtanovMS11}, the existence of an algorithm for \textsc{Dual Weighted Coloring} running in time $2^{o(k)} \cdot n^{\Ocal(1)}$ would
imply, in particular, an algorithm for \textsc{$3$-Coloring} running in time $2^{o(k)} \cdot n^{\Ocal(1)} = 2^{o(n)}$, contradicting the \ETH.

Our main contribution concerns the existence of (polynomial) kernels for \textsc{Dual Weighted Coloring}. By the well-known equivalence of admitting an \FPT algorithm and a kernel (cf.~\cite{DF13,CyganFKLMPPS15}), the \FPT algorithm mentioned above directly yields a kernel for \textsc{Dual Weighted Coloring} of size at most $9^k$. On the one hand, we considerably improve this bound by providing a kernel with at most $(2^{k-1}+1) (k-1)$ vertices, inspired by an approach attributed to Jan Arne Telle~\cite[Exercise 4.12.12]{DF13} for obtaining a quadratic kernel for  \textsc{Dual Vertex Coloring}.
On the other hand, we complement this result by showing that, unlike the \textsc{Dual Vertex Coloring} problem, \textsc{Dual Weighted Coloring} does {\sl not} admit a polynomial kernel unless ${\sf NP} \subseteq {\sf coNP} / {\sf poly}$, even on split graphs with only two different weights. We prove this result by a  polynomial parameter transformation from the \textsc{Set Cover} problem parameterized by the size of the universe, proved not to admit polynomial kernels unless ${\sf NP} \subseteq {\sf coNP} / {\sf poly}$ by Dom et al.~\cite{DomLS14}. Our reduction is an appropriate modification of a reduction of Demange et al.~\cite{DWMP02} to prove the \NP-hardness of \textsc{Weighted Coloring} on split graphs.

Motivated by the above hardness result, it is natural to identify graph classes on which the \textsc{Dual Weighted Coloring} problem admits polynomial kernels. We prove that this is the case of graph classes with bounded clique number and of interval graphs, for which we present a linear and a cubic kernel, respectively.  Finally, we identify subclasses of split graphs admitting polynomial kernels. Namely, we prove that \textsc{Dual Weighted Coloring} restricted to split graphs where each vertex in the clique has at most $d$ non-neighbors in the stable set, for some constant $d \geq 2$, admits a kernel with at most
$k^{d}$ vertices. We show that the dependency on $d$ in the exponent is necessary, by proving that  for any $\varepsilon > 0$, a kernel with $\Ocal(k^{\frac{d-3}{2}-\varepsilon})$ vertices on that graph class does not exist unless ${\sf NP} \subseteq {\sf coNP} / {\sf poly}$. In other words, we rule out the existence of a {\sl uniform} kernel, that is, a kernel of size $f(d) \cdot k^{\Ocal(1)}$ for any function $f$.

\medskip
\noindent\textbf{Organization of the paper.} In Section~\ref{sec:prelim} we present some basic preliminaries about graphs and parameterized complexity. In Section~\ref{sec:FPT-algo} we present the \fpt algorithm and in Section~\ref{sec:kernels} we provide the kernelization results. Finally, we conclude the article in Section~\ref{sec:further}.

\section{Preliminaries}
\label{sec:prelim}

\noindent \textbf{Graphs.} We use standard graph-theoretic notation, and we consider simple undirected graphs without loops or multiple edges; see~\cite{Die10} for any undefined terminology. Given a graph $G=(V,E)$, $X \subseteq V$, and $v \in V$, we denote $N_X(v) = N(v) \cap X$, where $N(v) = \{u \in V \mid \{u,v\} \in E\}$, and $\overline{N_X}(v) = X \setminus \{N_X(v) \cup \{v\}\}$. An \emph{antimatching} in a graph $G$ is a matching in the complement of $G$. Two vertices $u,v \in V(G)$ are \emph{twins} if $N(u) = N(v)$. A vertex $v \in V(G)$ is \emph{universal} if $N(v) = V(G) \setminus \{v\}$. A graph $G$ is a \emph{split graph} if $V(G)$ can be partitioned into an independent set and a clique, and an \emph{interval graph} if one can associate a real interval with each vertex, so that two vertices are adjacent if and only if the corresponding intervals intersect.


\medskip

\noindent \textbf{Parameterized complexity.} We refer the reader to~\cite{DF13,CyganFKLMPPS15} for basic background on parameterized complexity, and we recall here only some basic definitions, with special emphasis on tools for polynomial kernelization.
A \emph{parameterized problem} is a language $L \subseteq \Sigma^* \times \mathbb{N}$.  For an instance $I=(x,k) \in \Sigma^* \times \mathbb{N}$, $k$ is called the \emph{parameter}. 
A parameterized problem is \emph{fixed-parameter tractable} ({\sf FPT}) if there exists an algorithm $\Acal$, a computable function $f$, and a constant $c$ such that given an instance $I=(x,k)$,
$\Acal$ (called an {\sf FPT} \emph{algorithm}) correctly decides whether $I \in L$ in time bounded by $f(k) \cdot |I|^c$.


%

A fundamental concept in parameterized complexity is that of \emph{kernelization}. A kernelization
algorithm, or just \emph{kernel}, for a parameterized problem $\Pi $ takes an
instance~$(x,k)$ of the problem and, in time polynomial in $|x| + k$, outputs
an instance~$(x',k')$ such that $|x'|, k' \leqslant g(k)$ for some
function~$g$, and $(x,k) \in \Pi$ if and only if $(x',k') \in \Pi$. The function~$g$ is called the \emph{size} of the kernel and may
be viewed as a measure of the ``compressibility'' of a problem using
polynomial-time preprocessing rules. A kernel is called \emph{polynomial} (resp. \emph{linear}) if the function $g(k)$ is a polynomial (resp. linear) function in $k$.
A breakthrough result of Bodlaender et al.~\cite{BodlaenderDFH09} gave the first framework for proving that certain parameterized problems
do not admit polynomial kernels, by establishing so-called \emph{composition algorithms}. Together with a result of Fortnow and
Santhanam~\cite{FortnowS11} this allows to exclude polynomial kernels under the assumption that ${\sf NP} \nsubseteq {\sf coNP} / {\sf poly}$, otherwise implying
a collapse of the polynomial hierarchy to its third level~\cite{Yap83}. Very successful notions for proving such type of result are those of \emph{cross-composition}, introduced by Bodlaender et al.~\cite{BodlaenderJK14},  and of \emph{polynomial parameter transformation}, introduced by Bodlaender et al.~\cite{BodlaenderTY11}. We need to define the latter. A polynomial parameter transformation from a parameterized problem $P$ to a parameterized problem $Q$ is an algorithm that, given an instance $(x,k)$ of $P$, computes in polynomial time an equivalent instance $(x',k')$ of $Q$ such that $k'$ is bounded by a polynomial depending only on $k$.

\smallskip
Within parameterized problems, the class {\sf W}[1] may be seen as the parameterized equivalent to the class \np of classical optimization problems. Without entering into details (see~\cite{DF13,CyganFKLMPPS15} for the formal definitions), a parameterized problem being {\sf W}[1]-\emph{hard} can be seen as a strong evidence that this problem is {\sl not} \fpt.
The class {\sf W}[2] of parameterized problems is a class that contains $\W$[1], and such that the problems that are {\sf W}[2]-\emph{hard} are  even more unlikely to be \fpt than those that are {\sf W}[1]-hard (again, see~\cite{DF13,CyganFKLMPPS15} for the formal definitions).


\section{\FPT algorithm}
\label{sec:FPT-algo}

In this section we present an \FPT algorithm for the \textsc{Dual Weighted Coloring} problem.

\begin{theorem}\label{thm:FPT}
The \textsc{Dual Weighted Coloring} problem can be solved in time $9^k \cdot n^{\Ocal(1)}$.
\end{theorem}
\begin{proof}
We start by computing a maximum unweighted antimatching $\bar{M}$ in $G$; note that this can be done in polynomial time by computing a maximum matching in the complement of $G$. If $|\bar{M}| \geq k$, since we assume that all the vertex weights are at least 1, by putting in the same color class each pair of vertices that belong to a non-edge of $\bar{M}$ and coloring any other vertex with a new color, we obtain a coloring of $G$ with weight at most $\sum_{v \in V(G)}w(v) - k$, and we can output a constant-sized {\sf yes}-instance. Thus, we assume henceforth that $|\bar{M}| \leq k-1$.

Let $V(\bar{M}) \subseteq V(G)$ be the set of vertices that appear in the non-edges in $\bar{M}$. Since $\bar{M}$ is maximum, the set of vertices $K:= V(G) \setminus V(\bar{M})$ induces a clique in $G$. Note that, in any proper coloring $c$ of $G$, at least $|K|$ colors will be needed.
Let $K = \{v_1, \ldots, v_{|K|}\}$ and let $c(v_i) = c_i$.
Hence, it remains just to color the vertices in $V(\bar{M})$, which may be colored with some colors previously used in $K$ or with new ones.


Let $X \subseteq V(\bar{M})$ and $0 \leq i \leq |K|$. We define
$T(X, i)$ as the minimum weight of a coloring of $G[K \cup X]$ such that no color $c_j$, for $j > i$, is assigned to a vertex of $X$, i.e., the colors assigned to vertices in $X$ are either from the set $\{c_1, \ldots, c_i\}$ or new colors not assigned to any vertex of $K$. Note that, by definition, $\sigma(G,w) = T(V(\bar{M}),|K|)$.

We now describe how to compute $T(X, i)$ for every $X \subseteq V(\bar{M})$ and every $0 \leq i \leq |K|$.

If $X = \emptyset$, then, for any $i$,
$$T(\emptyset, i)= \sum_{v\in K} w(v),$$
since $K$ is a clique.
This can be done in linear time.

If $i = 0$, then the colors used in $K$ and in $X$ are disjoint. By applying brute force over all possible stable sets of $X$ we get
$$T(X, 0) = \min_{\substack{\emptyset \not =S \subseteq X \\ S\textrm{ stable}}} T(X \setminus S, 0) + w(S),$$
where $w(S) = \max_{v \in S}w(v)$. The values $T(X, 0)$, for all possible sets $X \subseteq V(\bar{M})$, can be computed in time $$\left(\sum_{j=0}^{|V(\bar{M})|} {|V(\bar{M})| \choose j}\cdot 2^j \right) \cdot n^{\Ocal(1)}= 3^{|V(\bar{M})|}\cdot n^{\Ocal(1)},$$
corresponding to considering the sets $X$ by increasing size and choosing an arbitrary subset $S$ inside $X$, and where the polynomial factor comes from checking whether such a  set $S$ is stable or not.

Now, if $i > 0$ and $X \not= \emptyset$, we have two possibilities, namely either color $c_i$ is used in $V(\bar{M})$ or not. If $c_i$ is not used, clearly we have $T(X, i) = T(X, i-1)$. Otherwise, there is a non-empty set $S \subseteq X$ of vertices colored with color $c_i$. Therefore, in the latter possibility we can iterate over every possible nonempty set $S$ and compute $T(X,i)$ as follows.
$$ T(X, i) = \min_{\substack{\emptyset \not =S \subseteq X \\ S \cup \{v_i\} \textrm{ stable}}} T(X \setminus S, i-1) + w(S \cup \{v_i\}) - w(v_i).
$$
Here we have to subtract the weight of $v_i$, which was, in the partial solution $T(X \setminus S, i-1)$, the weight of color $c_i$, and replace it by $w(S \cup \{v_i\})$, which may be greater. As in the previous case, for every $i$ it is possible to compute $T(X,i)$ for every $X$ in time $3^{|V(\bar{M})|} \cdot n^{\Ocal(1)}$.

Hence, since $|V(\bar{M})| \leq 2k-2$,  in time bounded by $3^{2k} \cdot n^{\Ocal(1)} = 9^{k} \cdot n^{\Ocal(1)}$ we can compute $T(V(\bar{M}),|K|) = \sigma(G,w)$ and answer whether $\sigma(G,w) \leq \sum_{v \in V(G)}w(v) - k$ or not.
\end{proof}

\section{Kernelization results}
\label{sec:kernels}

In this section we focus on the existence of (polynomial) kernels for \textsc{Dual Weighted Coloring}.

\begin{theorem}\label{thm:kernel}
The \textsc{Dual Weighted Coloring} problem admits a kernel with at most $(2^{k-1}+1) \cdot (k-1)$ vertices.
\end{theorem}
\begin{proof}
We start with  the following trivial polynomial-time reduction rule.

\medskip
\noindent\textbf{Rule 1}. If $G$ contains a universal vertex, delete it.

\begin{claim}\label{claim:rule1-safe}
Rule~1 is safe.
\end{claim}
\begin{proof} Since a universal vertex $u$ appears as a singleton in any proper coloring of $G$, it follows that $\sigma(G,w) \leq \sum_{v \in V(G)}w(v) - k$ if and only if $\sigma(G - \{u\},w) \leq \sum_{v \in V(G)\setminus \{u\}}w(v) - k$.
\end{proof}

As in the proof of Theorem~\ref{thm:FPT}, we compute in polynomial time a maximum unweighted antimatching $\bar{M}$ in $G$. 
Again, if $|\bar{M}| \geq k$, we can output a constant-sized {\sf yes}-instance, so we assume henceforth that $|\bar{M}| \leq k-1$. Let again $V(\bar{M}) \subseteq V(G)$ be the set of vertices that appear in the non-edges in $\bar{M}$, and recall that since $\bar{M}$ is maximum, the set of vertices $K= V(G) \setminus V(\bar{M})$ induces a clique in $G$.

We now partition $K$ into a set of equivalence classes $\mathcal{C}$ according to the neighborhood in $V(\bar{M})$. That is, $u,v \in K$ belong to the same class in $\mathcal{C}$ if and only if $N_{V(\bar{M})}(u) = N_{V(\bar{M})}(v)$. Note that $\mathcal{C}$ can be constructed in polynomial time, by iteratively processing the vertices of $K$, comparing the neighborhood in $V(\bar{M})$ of the currently processed vertex $v$ with those of the already processed vertices, and creating a new class containing only $v$ if no class has exactly the set $N_{V(\bar{M})}(v)$ as neighbors in $V(\bar{M})$.
Given an equivalence class $C \in \mathcal{C}$ and a non-edge  $\bar{e} \in \bar{M}$, we denote by $N_{\bar{e}}(C)$ the set of neighbors of any vertex in $C$ in  the set consisting of the two endpoints of $\bar{e}$. We proceed to analyze the number and the size of the classes in $\mathcal{C}$. We start with an easy claim.

\begin{claim}\label{claim:1}
Let $C \in \mathcal{C}$ be an equivalence class with $|C| \geq 2$. For every non-edge $\bar{e} \in \bar{M}$ it holds that $|N_{\bar{e}}(C)| \geq 1$.
\end{claim}
\begin{proof}
Suppose for contradiction that $|N_{\bar{e}}(C)| = 0$,  let $x,y$ be the endvertices of $\bar{e}$, and let $u,v$ be any two vertices in $C$. Then we can obtain from $\bar{M}$ a larger antimatching $ \bar{M}'$ by replacing the non-edge $\{x,y\}$ with the two non-edges  $\{u,x\}$ and $\{v,y\}$, a contradiction.
\end{proof}

In the next claim we restrict further the neighborhoods of the equivalence classes in $V(\bar{C})$.

\begin{claim}\label{claim:3}
Let $C_1,C_2 \in \mathcal{C}$ be two equivalence classes and let $\bar{e} \in \bar{M}$ be a non-edge with endpoints $x$ and $y$. Then it is not possible that $x \notin N_{\bar{e}}(C_1)$ and $y \notin N_{\bar{e}}(C_2)$, or vice versa.
\end{claim}
\begin{proof}
Suppose for contradiction that $x \notin N_{\bar{e}}(C_1)$ and $y \notin N_{\bar{e}}(C_2)$, and let $u \in C_1$ and $v \in C_2$. Then we can obtain from $\bar{M}$ a larger antimatching $ \bar{M}'$ by replacing the non-edge $\{x,y\}$ with the two non-edges  $\{u,x\}$ and $\{v,y\}$, a contradiction.
\end{proof}

We call an equivalence class $C \in \mathcal{C}$ \emph{special} if for some non-edge  $\bar{e} \in \bar{M}$, $|N_{\bar{e}}(C)| = 0$, and we call such an $\bar{e}$ an \emph{special} non-edge; otherwise we call an equivalence class \emph{normal}. We call a non-edge $\bar{e} \in \bar{M}$ \emph{normal} if for every $C \in \mathcal{C}$,  $|N_{\bar{e}}(C)| \geq 1$. Let $k_{\sf s}$ and $k_{\sf n}$ be the number of special and normal non-edges in $\bar{M}$, respectively, and note that $k_{\sf s} + k_{\sf n} = |\bar{M}| \leq k-1$.

By Claim~\ref{claim:1}, every special class contains exactly one vertex. If $\bar{e} \in \bar{M}$ is a special non-edge and $C \in \mathcal{C}$ is such that $|N_{\bar{e}}(C)| = 0$, then Claim~\ref{claim:3} implies that for every other class $C' \in \mathcal{C}$ different from $C$, it holds that $|N_{\bar{e}}(C')| = 2$. Therefore, the number of special classes in $\mathcal{C}$ is at most the number of special  non-edges in $\bar{M}$, that is, at most $k_{\sf s}$.

Let $\bar{e} \in \bar{M}$ be a normal non-edge with endpoints $x,y$, let $C \in \mathcal{C}$ be such that $|N_{\bar{e}}(C)| = 1$, and assume that $N_{\bar{e}}(C) = \{x\}$. Then Claim~\ref{claim:3} implies that for every class $C' \in \mathcal{C}$ such that $|N_{\bar{e}}(C')| = 1$, it holds that $N_{\bar{e}}(C') = \{x\}$. Hence, every normal non-edge $ \bar{e} \in \bar{M}$ has at most one endpoint that has some non-neighbor in $K$; we call such a vertex the \emph{missing} vertex of $ \bar{e}$. This means that for every normal class $C \in \mathcal{C}$ and every non-edge $ \bar{e} \in \bar{M}$, there are {\sl exactly two} possibilities: either $|N_{\bar{e}}(C)| = 2$, or $|N_{\bar{e}}(C)|=1$ and $N_{\bar{e}}(C)$ consists of the endpoint of  $\bar{e}$ distinct from its missing vertex. Moreover, by Claim~\ref{claim:3} the latter case can only occur if $\bar{e}$ is a normal non-edge. Therefore, the number of normal classes in $\mathcal{C}$ is at most $2^{k_{\sf n}}-1$, corresponding to all the choices of neighborhoods in the set consisting of the missing vertices in the normal edges in $\bar{M}$, and excluding the class $C$ with $N_{V(\bar{M})}(C) = V(\bar{M})$, since these vertices would be deleted by Rule~1.

Finally, in order to bound the size of the equivalence classes in $\mathcal{C}$, we state the following reduction rule, which says that it is enough to keep, for each equivalence class, the $|\bar{M}|$ heavier vertices.

\medskip
\noindent\textbf{Rule 2}. Suppose there exists an equivalence class $C \in \mathcal{C}$  with $|C| > |\bar{M}|$. Let $W \subseteq C$ be a subset of vertices with $|W| = |\bar{M}|$ and such that, if $u \in W$ and $v \in C \setminus W$, then $w(u) \geq w(v)$. Delete from $G$ all the vertices in $C \setminus W$.

\begin{claim}\label{claim:2}
Rule~2 is safe.
\end{claim}
\begin{proof} We need to prove that if $(G',w,k)$ results from $(G,w,k)$ after the application of Rule~2, then $(G,w,k)$ and $(G',w,k)$ are equivalent instances of \textsc{Dual Weighted Coloring}.

Assume first that $\sigma(G',w) \leq \sum_{v \in V(G')}w(v) - k$, and let $c'$ be a coloring of $(G',w)$ achieving this bound. We define a coloring $c$ of $(G,w)$ starting from $c'$ and creating a new color for each of the vertices in $C \setminus W$. Clearly, $w(c) = w(c') + \sum_{v \in C \setminus W}w(v) \leq \sum_{v \in V(G)}w(v) - k$.

Conversely, assume that $\sigma(G,w) \leq \sum_{v \in V(G)}w(v) - k$, and let $c$ be a coloring of $(G,w)$ achieving this bound. Due to Rule~1, we can clearly assume that $|\bar{M}| \geq 1$. Hence, since Rule~2 has been applied on $C \in \mathcal{C}$, it follows that $|C| \geq 2$. Thus, by Claim~\ref{claim:1}, we have that for every non-edge $\bar{e} \in \bar{M}$, $|N_{\bar{e}}(C)| \geq 1$. This implies, together with the fact that the vertices in $C$ induce a clique, that in the coloring $c$ of $(G,w)$, at most $|\bar{M}|$ vertices of $C$ appear in colors containing vertices of $V(\bar{M})$, and every other vertex in $C$ is a singleton in its color. Let $T \subseteq C$ be this set of at most $|\bar{M}|$ vertices, and let $W \subseteq C$ be the set such that
Rule~2 has deleted from $G$ the vertices in $C \setminus W$. We iteratively modify the coloring $c$ by updating the set $T$ as follows (the set $W$ remains the same). While there exists a vertex $u \in V(G)$ such that $u \in T$ and $u \notin W$, let $v \in W$ such that $v \notin T$ (note that vertex $v$ exists, since $|T| \leq |\bar{M}| = |W|$), and swap $u$ and $v$ in the coloring, that is, now $u$ is a singleton in its color and $v$ is in a color containing vertices of $V(\bar{M})$. Since $u$ and $v$ are twins, this procedure indeed creates a proper coloring of $G$, which we also call $c$ with abuse of notation. Let us now argue that the weight of the coloring has not increased. Note that since $u \notin W$ and $v \in W$, it follows that $w(v) \geq w(u)$. If we denote by $S_i$ and $S_j$ the colors of $c$ before the swapping, so that $u \in S_i$ and $ S_j = \{v\}$, and by $S_i'$ and $S_j'$ the corresponding colors after the swapping, note that $w(S_i) \geq w(u)$, $w(S_j) = w(v)$, $w(S_i') = \max\{w(S_i),w(v)\}$, and $w(S_j') = w(u)$. Therefore,
$$
w(S_i') + w(S_j') = \max\{w(S_i),w(v)\} + w(u) \leq w(S_i) + w(v) = w(S_i) + w(S_j),
$$
where we have used that $w(v) \geq w(u)$. Thus, at the end of this procedure we obtain a coloring $c$ with $w(c) \leq \sum_{v \in V(G)}w(v) - k$ such that every vertex in $C \setminus W$ is a singleton in its color class. We define $c'$ to be the restriction of $c$ to $G'$. Since every vertex in $C \setminus W = V(G) \setminus V(G')$ is a singleton in its color class in $c$, it follows that
$$
w(c') = w(c) - \sum_{v \in C \setminus W}w(v) \leq \Big(\sum_{v \in V(G)}w(v) - k\Big) - \sum_{v \in C \setminus W}w(v) = \sum_{v \in V(G')}w(v) - k,
$$
and the claim follows.
\end{proof}

We can easily apply Rule~2 exhaustively in polynomial time to all the classes $C \in \mathcal{C}$  with $|C| > |\bar{M}|$. We call an instance \emph{reduced} if none of Rule~1 and Rule~2 can be applied anymore. The above discussion implies that if $(G,w,k)$ is a reduced instance, then
$$
|V(G)| = |V(\bar{M})| + |K| \leq 2(k-1) + k_{\sf s} + (2^{k_{\sf n}}-1)\cdot (k-1).
$$
Using that $k_{\sf s} + k_{\sf n} \leq k-1$, from the above equation we get that
$$
|V(G)| \leq 2(k-1) + (2^{k-1}-1)\cdot (k-1) = (2^{k-1}+1) \cdot (k-1),
$$
and the theorem follows.
\end{proof}

It is worth mentioning that the analysis of the kernel size in the proof of Theorem~\ref{thm:kernel} is {\sl tight}. Indeed, let $G$ consist of an antimatching $\bar{M}$ of size $k-1$, and let $K$ consist of $2^{k-1}-1$ equivalence classes with $k-1$ vertices, each having a distinct non-complete neighborhood in the set consisting of one (arbitrary) vertex of each non-edge in $\bar{M}$. One can easily check that $|V(G)| = (2^{k-1}+1) \cdot (k-1)$, that none of Rule~1 and Rule~2 can be applied to $G$, and that $G$ contains no antimatching strictly larger than $k-1$.

We complement the result of Theorem~\ref{thm:kernel} by showing that, unless ${\sf NP} \subseteq {\sf coNP} / {\sf poly}$, the problem does not admit polynomial kernels.

\begin{theorem}\label{thm:no-kernel}
The \textsc{Dual Weighted Coloring} problem does not admit a polynomial kernel unless ${\sf NP} \subseteq {\sf coNP} / {\sf poly}$, even on split graphs with only two different weights.
\end{theorem}
\begin{proof}
We present a polynomial parameter transformation from the \textsc{Set Cover} problem parameterized by the size of the universe; Dom et al.~\cite{DomLS14} proved that this problem does not admit a polynomial kernel unless ${\sf NP} \subseteq {\sf coNP} / {\sf poly}$.
Our reduction is almost the same as the reduction of Demange et al.~\cite{DWMP02} from \textsc{Set Cover} to \textsc{Weighted Coloring} on split graphs, only the vertex weights change. Let $(U, \mathcal{S}, k, \ell)$ be an instance of \textsc{Set Cover}, where $\mathcal{S}$ is a family of sets of elements over a universe $U$ of size $k$, and the question is whether there exists a subset $\mathcal{S}' \subseteq \mathcal{S}$ of at most $\ell$ sets covering all the elements of $U$. We construct an instance $(G,w,k')$ of  \textsc{Dual Weighted Coloring} as follows. The graph $G$ contains a clique $K$ on $|\mathcal{S}|$ vertices and an independent set $I$ on $k$ vertices. The vertices of $K$ and $I$ are associated, respectively, with the sets in $\mathcal{S}$ and the elements in $U$. There is an edge between a vertex in $K$ and a vertex in $I$ if and only if the corresponding set does {\sl not} contain the corresponding element. All the vertices in $K$ have weight $\ell$,  and all the vertices in $I$ have weight $\ell + 1$. Note that $G$ is indeed a split graph with only two different weights. Finally, we set $k' = k(\ell + 1) - \ell$. Since  we can clearly assume that $\ell \leq k$, as otherwise the instance is trivial, it follows that $k' = \Ocal(k^2)$, which is required in a polynomial parameter transformation. We claim that $(U, \mathcal{S}, k, \ell)$ is a {\sf yes}-instance of \textsc{Set Cover} if and only if $\sigma(G,w) \leq \sum_{v \in V(G)}w(v) - k'$.

Assume first that $(U, \mathcal{S}, k, \ell)$ is a {\sf yes}-instance, and let
$\mathcal{S}' \subseteq \mathcal{S}$ be a solution with $|\mathcal{S}' | \leq \ell$. We define a coloring $c$ of $G$ as follows. We start with a color for each vertex in $K$ and, for every element of $U$, we include its corresponding vertex of $I$ into one of the colors corresponding to the sets in $\mathcal{S}'$ containing that element. One can easily check that $ w(c) \leq |\mathcal{S}| \cdot \ell + \ell = \sum_{v \in V(G)}w(v) - k'$.

Conversely, assume that $\sigma(G,w) \leq \sum_{v \in V(G)}w(v) - k' = |\mathcal{S}| \cdot \ell + \ell$, and let $c$ be a coloring of $(G,k)$ achieving this bound. Since $\sum_{v \in k}w(v) = |\mathcal{S}| \cdot \ell$ and the vertices in $I$ have weight $\ell + 1$, all the vertices in $I$ have to be included in at most $\ell$ out of the $|\mathcal{S}|$ colors of $c$ containing the vertices of the clique $K$. By construction of $G$, this is possible only if there exists a subset of at most $\ell$ sets in $\mathcal{S}$ covering all the elements of $U$, and the theorem follows.
\end{proof}

In view of Theorem~\ref{thm:no-kernel}, in what follows we focus on identifying graph classes on which the \textsc{Dual Weighted Coloring} problem admits a polynomial kernel.

\begin{remark}
The problem clearly admits a kernel of size $\Ocal(k)$ on sparse graphs, since if there are no large cliques, then the clique $K$ defined in the proof of Theorem~\ref{thm:no-kernel} (that is, the complement of the maximum antimatching) is of constant size. More formally, if $\omega$ is the maximum clique size of a graph in the class, then we get a kernel with at most $2k-2+\omega$ vertices.
\end{remark}

In our next result we provide a polynomial kernel on a relevant class of dense graphs, namely that of interval graphs.

\begin{proposition}\label{prop:kernel-interval}
The \textsc{Dual Weighted Coloring} problem restricted to interval graphs admits a kernel with at most $k^3-2k^2+2k-1$ vertices.
\end{proposition}
\begin{proof}
We will proceed as in the proof of Theorem~\ref{thm:kernel} and show that, if the input graph $G$ is an interval graph, then the number of equivalence classes is quadratic in the parameter.  As before, we assume that $\bar{M}$ is a maximum antimatching of $G$, that $|\bar{M}| \leq k-1$, and also that Rule~1 and Rule~2 have been exhaustively applied. We also consider $K = V(G) \setminus V(\bar{M})$ and, as before, note that $K$ induces a clique.

We will show that the number of maximal cliques of $G$ is bounded by a linear function of $|\bar{M}|$. For that, we will make use of a well-known result of Fulkerson and Gross~\cite{FG65} stating that a graph $G$ is an interval graph if and only if the 0/1 incidence matrix $\mathcal{M}$ of vertices and maximal cliques of $G$ has the \emph{consecutive ones property}, i.e., the columns of $\mathcal{M}$ can be permuted so that the ones in each row appear consecutively. A consequence of this result is that the set $\mathcal{C}$ of maximal cliques of $G$ can be arranged as $\{C_1,\ldots, C_p\}$ in such a way that, for each vertex $v$, there are indices $\ell_v$ and $r_v$ such that $v \in C_i$ if and only if $\ell_v \leq i \leq r_v$.

Note that the incidence of a vertex to the maximal cliques completely defines its neighborhood. This implies
that $G$ has at most $p \choose 2$ neighborhood classes. Hence, for obtaining a cubic kernel it suffices to show that $p = \Ocal(k)$. In what follows we  provide explicit bounds on both $p$ and the size of the kernel; we will show later that these bounds are tight.

\begin{claim}\label{claim:interval}
Let $G$ be an interval graph with $p \geq 2$ maximal cliques and let $\bar{M}$ be a maximum antimatching of $G$. Then $p \leq 2|\bar{M}|$.
\end{claim}
\begin{proof}
We proceed by induction on $p$. Clearly, if $p = 2$, then there are at least two non-adjacent vertices, so $|\bar{M}| \geq 1$  and $p \leq 2|\bar{M}|$.

Now let $\{C_1,\ldots, C_p\}$ be an ordering of the maximal cliques of $G$ as defined before. Since the $C_i$'s are maximal cliques, for any two consecutive cliques $C_i$ and $C_{i+1}$ it holds that $C_i \setminus C_{i+1} \not= \emptyset$ and $C_{i+1} \setminus C_{i} \not= \emptyset$. In particular, together with the consecutive ones property, this implies that $C_1$ has an exclusive vertex $v$, that is, a vertex that does not belong to any maximal clique other than $C_1$. Let $u \in C_2 \setminus C_1$ and $G' = G[V(G) \setminus \{u,v\}]$.

Let $\bar{M}'$ be a maximum antimatching of $G'$ and $p'$ be the number of maximal cliques of $G'$. It is easy to see that $1 \leq |\bar{M}| - |\bar{M}'| \leq 2$, since $u$ and $v$ are not adjacent and can extend $\bar{M}'$, and their removal from $G$ can destroy at most two non-edges from $\bar{M'}$. By induction, it holds that
$$p' \leq 2|\bar{M}'| \leq 2(|\bar{M}| - 1) = 2|\bar{M}| - 2.$$
Now it suffices to show that $p - 2\leq p'$. Let $r \leq p$ be the largest integer such that $u \in C_r$. For $1 \leq i \leq p$, let $C_i' = C_i \setminus \{u\}$. We will show that for each $i,j \in [2, p]\setminus \{r\}$, with $i < j$, $C_i'$ and $C_j'$ are incomparable maximal cliques in $G'$, implying that $p' \geq p-2$.

For the incomparability, note that, as observed above,
$C_i \setminus C_{i+1} \not= \emptyset$, and let $v_i \in C_i \setminus C_{i+1}$. Since $i\not=r$, it follows that $v_i \not=u$, and the consecutive ones property implies that $v_i \not\in C_j$ and then $v_i \in C_i' \setminus C_j'$. The fact that $C_j' \setminus C_i' \not= \emptyset$ can be proved with a symmetric argument.

For the maximality, suppose for contradiction that $C$ is a maximal clique in $G'$ having $C_i'$ as a proper subset for some $i \in [2, p]\setminus \{r\}$. First note that $u \in C_i$, otherwise $C_i' = C_i$ and $C$ would contradict the maximality of $C_i$ in $G$. This implies that $2 \leq i < r$. Now, if for some $u' \in C \setminus C_i'$ it holds that $uu' \in E(G)$, then $C_i$ would not be a maximal clique of $G$, since $C_i \cup \{u'\}$ would be a clique in $G$. On the other hand, if every $u' \in C \setminus C_i'$ is nonadjacent to $u$ then either $C \setminus C_i' \subseteq C_1$, and then the vertices in $C \setminus C_i'$ are not adjacent to the vertices in $C_i \setminus C_{i-1}$, which is a contradiction since $C_i \setminus C_{i-1}\subseteq C_i \setminus \{u\} = C_i' \subseteq C$ and $C$ is a clique, or $C \setminus C_i' \subseteq C_{r+1} \cup \dots \cup C_p$, and then the vertices in $C \setminus C_i'$ are not adjacent to the vertices in $C_i \setminus C_{i+1}$, getting again a contradiction. This completes the proof of the claim.
\end{proof}

Since $K$ is a clique, there is an index $i$ such that $K \subseteq C_i$. For each $v \in K$, let $\ell_v$ be the smallest index such that $v \in C_{\ell_v}$ and $r_v$ be the largest index such that $v \in C_{r_v}$. Recall that $\ell_v$ and $r_v$ completely define the neighborhood of $v$.
Note that for all $v \in K$, $\ell_v \leq i \leq r_v$ and that either $\ell_v \not=1$ or $r_v \not=p$, since otherwise $v$ would be a universal vertex, which is impossible because we applied Rule 1 exhaustively. Hence, the maximum number of distinct possible combinations of indices $\ell_v$ and $r_v$ is $i \cdot (p-i+1) - 1$, which is maximized when $i = \left\lfloor\frac{p+1}{2}\right\rfloor$, giving a total number of
$$
\left\lfloor\frac{p+1}{2}\right\rfloor \cdot \left\lceil\frac{p+1}{2}\right\rceil - 1
$$
distinct equivalence classes. Using the fact that $p \leq 2|\bar{M}|$ by Claim~\ref{claim:interval}, and that $|\bar{M}| \leq k-1$, we get
$$
\left\lfloor\frac{p+1}{2}\right\rfloor \cdot \left\lceil\frac{p+1}{2}\right\rceil - 1 \leq
\left\lfloor\frac{2k-1}{2}\right\rfloor \cdot \left\lceil\frac{2k-1}{2}\right\rceil - 1 =
k(k-1) - 1.
$$
Since, by virtue of the application of Rule 2, each class has at most $(k-1)$ elements, we have
$$|V(G)| = |V(\bar{M})| + |V(K)| \leq
2(k-1) + (k-1)\cdot (k(k-1) - 1) =
k^3-2k^2+2k-1.$$
\end{proof}

Similarly to Theorem~\ref{thm:kernel}, we can show that the analysis of the kernel size in the proof of Proposition~\ref{prop:kernel-interval} is tight. Indeed, given an integer $k$, it is possible to build an interval graph with $2k-2$ maximal cliques, with one exclusive vertex in each clique, corresponding to vertices of $V(\bar{M})$, and $(k-1)$ vertices belonging exactly to cliques $C_i, \ldots, C_j$, for each
$$(i,j) \in
\left[1,\left\lfloor\frac{p+1}{2}\right\rfloor\right]\times
\left[\left\lfloor\frac{p+1}{2}\right\rfloor,\ 2k-2\right]
\setminus \{1,2k-2)\}.
$$
That interval graph attains the bound in the statement of Proposition~\ref{prop:kernel-interval} and cannot be reduced by Rule 1 or Rule 2.

\medskip

\begin{remark}
The result of Proposition~\ref{prop:kernel-interval} cannot be generalized to chordal graphs, as split graphs are chordal, and by Theorem~\ref{thm:no-kernel} the existence of a polynomial kernel on split graphs would imply that ${\sf NP} \subseteq {\sf coNP} / {\sf poly}$.
\end{remark}

\medskip


Our last result deals with a subclass of split graphs motivated by the fact that \textsc{Dual Weighted Coloring} on split graphs seems to have a close relation with the \textsc{Set Cover} problem.

\begin{proposition}\label{prop:kernel-split-bounded-degree}
The \textsc{Dual Weighted Coloring} problem restricted to split graphs such that each vertex in the clique has at most $d$ non-neighbors in the stable set, for some constant $d \geq 2$, admits a kernel with
at most $k^{d}$ vertices. Furthermore, for any $\varepsilon > 0$, a kernel with $\Ocal(k^{\frac{d-3}{2}-\varepsilon})$ vertices does not exist unless ${\sf NP} \subseteq {\sf coNP} / {\sf poly}$.
\end{proposition}
\begin{proof}
For the positive result, let $(G,w,k)$ be an instance of \textsc{Dual Weighted Coloring}, with $G$ being a split graph such that each vertex in the clique has at most $d$ non-neighbors in the stable set, for some integer $d \geq 2$. 
We mimic the proof of Theorem~\ref{thm:kernel}, and we slightly change the analysis. Recall that $|V(\bar{M})| \leq 2(k-1)$ and that the number of special classes in $\mathcal{C}$, each containing exactly one vertex, is at most $k_{\sf s}$. Concerning the normal classes in $\mathcal{C}$, since each vertex in the clique has at most $d$ non-neighbors in the stable set (hence, in the graph $G$ itself as well),  the number of distinct neighborhoods in the set consisting of the missing vertices in the normal edges in $\bar{M}$, which is an upper bound on the number of normal classes, is at most ${k_{\sf n} \choose d} \leq k_{\sf n}^d$.
Therefore,  if $(G,w,k)$ is a reduced instance, then
$$
|V(G)| = |V(\bar{M})| + |K| \leq 2(k-1) + k_{\sf s} + k_{\sf n}^d\cdot (k-1) \leq (k-1)^d + 2(k-1) \leq k^d.
$$
where in the last inequality we have used that $d \geq 2$.
%

For the negative result, we reuse the reduction of Theorem~\ref{thm:no-kernel}, but starting from the \textsc{$d$-Set Cover} problem, that is, the restriction of \textsc{Set Cover} to instances where each set contains at most $d$ elements. Hermelin and Wu~\cite{W12} proved that, for any fixed $d \geq 2$, \textsc{$d$-Set Cover} does not admit kernels of size $\Ocal(k^{d-3-\varepsilon})$ for any $\varepsilon > 0$, unless ${\sf NP} \subseteq {\sf coNP} / {\sf poly}$, where $k$ is the size of the solution. Nevertheless, in the hardness proof for \textsc{$d$-Set Cover} given in~\cite{W12}, the size of the universe of the constructed instance is equal to $kd$. Therefore, we can conclude that \textsc{$d$-Set Cover} does not admit kernels of size $\Ocal(k^{d-3-\varepsilon})$ for any $\varepsilon > 0$, unless ${\sf NP} \subseteq {\sf coNP} / {\sf poly}$, where $k$ is the size of the universe. Moreover, the results in~\cite{W12} also rule out the existence of a \emph{bikernel}, that is, a relaxed kernelization notion where the output instance is not necessarily of the same problem.


Given an instance $(U, \mathcal{S}, k, \ell)$ of \textsc{$d$-Set Cover}, where $k$ is the size of the universe, we construct an instance $(G,w,k')$ of  \textsc{Dual Weighted Coloring} as in the proof of Theorem~\ref{thm:no-kernel}. Note that $G$ is indeed a split graph such that each vertex in the clique has at most $d$ non-neighbors in the stable set, and recall that $k' = k(\ell+1) - \ell$. Since we may assume that $\ell \leq k$, it follows that $k' \leq k^2$. Assume for contradiction that \textsc{Dual Weighted Coloring} restricted to this type of instances admits a kernel with $\Ocal(k^{\frac{d-3}{2}-\varepsilon})$ vertices, for some $\varepsilon > 0$. Then the composition of the above reduction with such a kernel would yield a {\sl bikernel} for \textsc{$d$-Set Cover} of size $\Ocal(k^{d-3-\varepsilon})$ for some $\varepsilon > 0$, which is impossible by the results of~\cite{W12} unless ${\sf NP} \subseteq {\sf coNP} / {\sf poly}$.
\end{proof}

\section{Further research}
\label{sec:further}
In this article we investigated the dual parameterization of the \textsc{Weighted Coloring}  problem, and we provided several positive and negative results, especially concerning polynomial kernelization. It would be interesting to identify other classes of (dense) graphs on which the problem admits polynomial kernels. It remains to close the gap in the degree of the polynomial kernels on the subclasses of split graphs considered in Proposition~\ref{prop:kernel-split-bounded-degree}. Another question is whether the cubic  kernel on interval graphs given in Proposition~\ref{prop:kernel-interval} can be improved, even on {\sl proper} interval graphs. Finally, one could try to prove lower bounds under the  \SETH (see~\cite{LokshtanovMS11}) on the running time of any \FPT algorithm solving \textsc{Dual Weighted Coloring}, hopefully getting close to the running time given in Theorem~\ref{thm:FPT}.

\bibliographystyle{abbrv}
\bibliography{Biblio-WC}

\end{document}